\documentclass[12pt]{article}
 \usepackage{bbm}

\usepackage{natbib}
\usepackage{amssymb,latexsym}
\usepackage{graphicx,amsmath,amsfonts,amsthm}
\usepackage{comment}
\usepackage{color,soul}

\def\row#1#2{{#1}_1,\ldots ,{#1}_{#2}}

\input prepictex
\input postpictex
\input pictexwd

\newcounter{theexample} 

\newcounter{themyclaim}
\def\Q{\mathbb{Q}}

\def\L{\mathcal{L}}

\def\row#1#2{{#1}_1,\ldots ,{#1}_{#2}}

\def\2vec#1#2{\left(\begin{array}{c}{#1}\\{#2}\end{array}\right)}

\newtheorem{theorem}{Theorem}

\newtheorem{lemma}{Lemma}

\newtheorem{proposition}{Proposition}
\newtheorem{example}{Example}

\newtheorem{definition}{Definition}

\usepackage{multirow}
\usepackage{colortbl}
\usepackage[usenames,dvipsnames]{xcolor}
 
\begin{document}

\title{\bf The single-crossing property on a tree}
\author{A. Clearwater,  C. Puppe and A. Slinko
}
\date{}

\maketitle

\begin{abstract}
\begin{quote}
We generalize the classical single-crossing property to single-crossing property on trees and obtain new ways to construct Condorcet domains which are sets of linear orders which possess the property that every profile composed from those orders have transitive majority relation. We prove that for any tree there exist profiles that are single-crossing on that tree; moreover, that tree is minimal in this respect for at least one such profile. Finally, we provide a polynomial-time algorithm to recognize whether or not a given profile is single-crossing with respect to some tree. We also show that finding winners for Chamberlin-Courant rule is polynomial  for profiles that are single-crossing on trees. 
\end{quote}
\end{abstract}

\section{Introduction}

Condorcet's famous paradox 
demonstrates that pairwise majority voting may produce
intransitive collective preferences. The question whether, and if so, how, this problem can be overcome by means
of restrictions on the domain of admissible individual preferences has attracted constant interest over the recent
decades, see \cite{Gaertner2001} for a detailed overview. 
Probably the
best-known of these is single-peakedness \cite{Black1948} which is frequently employed in models of political economy. It stipulates that all alternatives can be arranged along one dimension, for instance according to the political left-right spectrum. 
The concept of single-peakedness itself can be generalized considerably, 
however, the sufficiency of single-peakedness for
transitivity of the (strict) majority relation is confined to the classical one-dimensional case.\footnote{Single-peakedness
on trees still guarantees the existence of a Condorcet winner, and single-peakedness on a median graph the existence
of a ``local'' Condorcet winner \cite{BandeltBarthelemy1984}.}

A different and frequently useful sufficient condition for transitivity of the majority relation is the following single-crossing
property. Suppose that {\em voters} can be arranged on a one-dimensional linear spectrum. A profile of individual preferences
is said to have the single-crossing property if, for all pairs of alternatives $(a,b)$, the set of voters who prefer $a$ to $b$
forms a convex set in the one-dimensional spectrum.  As shown by 
\cite{Roth91}, in every single-crossing profile there exists the so-called {\it representative voter}, i.e.,~a voter whose (strict) preference
coincides with the (strict) majority relation.\footnote{One can show that such result does not hold for single-peaked
preferences even on a line.}

\cite{Roberts1977} and \cite{GansSmart1996} provide a number of economic applications of single-crossingness. All of them represent situations when preferences of individuals depend on a single parameter. For example, in Roberts' seminal paper voters' preferences on the level of taxation depend solely on their income: the lower the income the higher taxation this individual prefers. However, if the state provides subsidies for families with children, this condition may not be satisfied.  But if we fix the income we will find single-crossing condition relative to the second parameter, i.e., the more children the person has the higher the level of taxation she prefers. So there are compelling economic reasons which prompt us to consider single-crossingness on graphs more general than a line. 

 In contrast to the case of single-peaked preferences, the sufficiency of the
single-crossing property for transitivity of the strict majority relation generalizes to a larger class of graphs. In the present
paper, we prove the transitivity of the strict majority relation and a representative voter theorem for single-crossing profiles on trees.
This result also follows from the analysis of  
\cite{Demange2012} based on the notion of {\em intermediate preferences}
\cite{Grandmont1978}. 

The second part of the paper is devoted to algorithmic aspects. First, we give a constructive proof of the existence of a
single-crossing profile for any tree with $n$ vertices. We prove that $n$ alternatives are always sufficient and that there
exist trees for which single-crossing profiles with less than $n$ alternatives do not exist. We also give a polynomial-time
algorithm that recognizes whether or not a given profile is single-crossing with respect to some tree. Finally, we prove that 
the Chanberlin-Courant multi-winner voting rule on single-crossing profiles on trees has a polynomial time winner-determination 
problem which generalises a similar result of \cite{SYFE2013} for classical single-crossing property. It is interesting to note that for 
the single-peaked property on a tree only the egalitarian version of the Chanberlin-Courant rule remains polynomial. The classical 
utilitarian version of this rule becomes NP-hard \cite{YuCE13}.

The problem addressed in the present paper is closely related to the search of so-called (maximal) ``Condorcet
domains'' \cite{AbelloJohnson84,Abello91,GR:2008,DKK:2012}; see also the survey on the topic in~\cite{Mon:survey}.
Indeed, any single-crossing profile on a tree provides us with a new type of Condorcet domain (although possibly
not maximal). 

\section{Preliminaries}


Let $A$ and $V$ be two finite sets of cardinality 
$m$ and $n$ respectively. The elements of $A$ will be called alternatives, the elements of
$V=\{1,2,\ldots, n\}$ voters. We assume that the voters have preferences over the set of alternatives.  By
$\L(A)$ we denote the set of all (strict) linear orders on $A$; they represent the preferences of
agents over $A$. The elements of the Cartesian product
$
\L(A)^n=\L(A)\times\ldots\times \L(A)\ \ \ \mbox{($n$ times)}
$
are called $n$-profiles or simply profiles. They represent the
collection of preferences of an
\mbox{$n$-element} society of voters $V$ over alternatives from $A$. If a linear order 
$R_i\in \L(A)$ represents the preferences of the $i$-th agent, then by
$aR_ib$, where  $a,b\in A$, we denote that this agent prefers $a$ to $b$. We also denote this as $a\succ_i b$.\par

\begin{definition}
Let $R=(\row Rn)$ be a profile. The {\em majority relation} $M(R)$ of $R$ over $A$  is the binary relation on $A$ such that for any 
$a,b\in A$ we have $a\succeq b$ if and only if
$|(\{i\mid aR_ib\}|\ge |(\{i\mid bR_ia \}|$.\par
\end{definition}

When $n$ is odd, the majority relation is a tournament on $A$, i.e., complete and asymmetric binary relation. When $n$ is even, we may have an indifference when both $a\succeq b$ and $b\succeq a$  which  we denote as $a\sim b$. We will also write $a\succ b $ if $a\succeq b$ but not $b\succeq a$ and call it the {\em strict majority relation}.  \par

\begin{definition}
A Condorcet Domain is a set of linear orders $C\subseteq \L(A)$ such that, no matter how many voters in the profile $P$ have each of the linear orders from $C$ as their preference relation,  the strict majority relation of $P$ is  transitive.
\end{definition}


A tree is a connected graph $T=(V,E)$ with the set  of vertices $V$ and the set of edges $E$ such that $|E|=|V|-1$. Since there is a unique path between two vertices $u$ and $v$ we can define the distance $d(u,v)$ between them as the number of edges in the unique path between $u$ and $v$.  A subgraph  $T_1=(V_1,E_1)$ is a subtree of $T$, if it is a tree in its own right.
Intersection of any two subtrees is a subtree. If any edge in a tree $T=(V,E)$ is removed, the tree becomes a union of two disconnected subtrees.  See \cite{Diestel2005} for these and further facts about trees.

\section{Condorcet domains related to trees}

\subsection{The single-crossing property on a tree}

This concept generalizes the classical definition of single-crossingness  \cite{Mirr71,GansSmart1996}.

\begin{definition}
Let $P$ be an $n$-voter profile with the set of voters $V$ and the set of alternatives $A$ and $T=(V,E)$ be a tree on the set of voters $V$. We say that the profile $P$ is {\em single-crossing} with respect to $T$ if for every pair of alternatives $a,b\in A$ one of the following holds:
\begin{itemize}

\item We can remove an edge $e=(i,j)$ from $T$ so that for the two resulting subtrees $T_1=(V_1,E_1)$ and $T_2=(V_2,E_2)$ all voters in $V_1$ rate $a$ above $b$ and all voters in $V_2$ rate $b$ above $a$. In this case the edge $e$ will be called {\em $ab$-cut}. 
\item All voters prefer $a$ to $b$ or  all voters prefer $b$ to $a$. In this case we say that the $ab$-cut is {\em virtual}.
\end{itemize}
\end{definition}

\noindent An $ab$-cut partitions the vertices $V=V_{ab}\cup V_{ba}$, where $V_{ab}=\{i\in N\mid a\succ_i b\}$ and $V_{ba}=\{i\in N\mid b\succ_i a\}$;  let $T_{ab}$ and $T_{ba}$ denote the corresponding subtrees of~$T$.

\begin{example}[Classical single-crossing profile]
Suppose the set of alternatives $A$ can be endowed with a linear order $>_A$ such that for any two alternatives  $a,b\in A$ and for any integers $i,j\in N$ with $i<j$
\[
a>_A b\quad \text{and}\quad a\succ_i b \Longrightarrow a\succ_jb.
\]
It is clear that for any pair of alternatives $a,b\in A$ with $a>_Ab$, as $k$ runs from 1 to $n$, the relation $b\succ_k a$ switches to $a\succ_k b$ at most once. Hence for every pair $a,b\in A$ we can find $k_0=k(a,b)\in \{0,1,\ldots,n\}$ such that we have $b\succ_k a$ for $k\le k_0$ and  $a\succ_k b$ for $k>k_0$. This means that such profile is single-crossing with respect to a linear graph
 
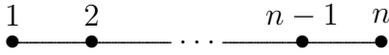
\begin{center}
\beginpicture
    \setcoordinatesystem units <0.7mm,0.7mm>
    \setplotarea x from -25 to 30, y from 0 to 10
      \setlinear
 \plot 0 0 15 0 30 0 /
    \put {$\bullet$} at 0 0
        \put {$\bullet$} at 15 0
\put {$1$} at 0 5
\put {$2$} at 15 5
\put{$\ldots$} at 35 0
 \plot 40 0 55 0 /
  \plot 55 0 70 0 /
 \put {$n-1$} at 55 5
  \put {$\bullet$} at 55 0
   \put {$n$} at 70 5
  \put {$\bullet$} at 70 0
\endpicture
\end{center}
\end{example}

Let us show that when the tree is not linear we can obtain some new single-crossing profiles which are not single-crossing in the classical sense. 

\begin{example}
\label{smallstar}
Let us consider the following profile, let us call it $P$, where voters 1,2,3,4 have preference orders $a\succ_1b\succ_1c\succ_1d$, $a\succ_2c\succ_2b\succ_2d$,  $d\succ_3a\succ_3c\succ_3b$, and $c\succ_4b\succ_4a\succ_4d$, respectively.
Obviously, this profile is not a single-crossing in the classical sense.  If it were, then voter 3 must be either the first or the last on the line since this is the only voter who ranks $d$ above $a$. Voter 4 also must be the first or the last since she is the only voter who ranks $c$ above $a$. Finally voter 1 must be either first or the last since she is the only one who ranks $b$ above $c$.

However, it is easy to check that $P$ is single-crossing with respect to the following graph with vertices associated with voters in the following way:
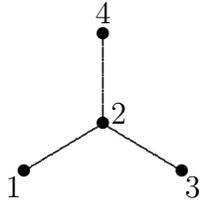
\begin{center}
\beginpicture
    \setcoordinatesystem units <0.7mm,0.7mm>
    \setplotarea x from -45 to 30, y from 0 to 30
      \setlinear
 \plot 0 0 15 9 30 0 /
  \plot 15 9 15 26 /
    \put {$\bullet$} at 0 0
        \put {$\bullet$} at 15 9
            \put {$\bullet$} at 15 26
                \put {$\bullet$} at 30 0
  \put {$2$} at 18 11
    \put {$1$} at -2 -3
        \put {$3$} at 32 -3
            \put {$4$} at 15 30
\endpicture
\end{center}
Indeed, for pairs $(a,b)$ and $(a,c)$ we have to cut the edge $(2,4)$. For pair $(b,c)$ we have to cut $(1,2)$ and for pairs $(a,d)$, $(b,d)$ and $(c,d)$ we have to cut $(2,3)$. 
\end{example} 


The following generalisation of the classical Representative Voter Theorem \cite{Roth91} follows from \cite{Demange2012} but we give a direct proof.

\begin{theorem}[Representative Voter Theorem]
\label{RVT}
Let $n$ be odd. If a profile $P=(\row Pn)$  is single-crossing with respect to a tree $T=(V,E)$, then there exists $i\in \{1,\ldots,n\}$ such that    preference order $P_i$  coincides with  the majority relation. 
\end{theorem}

\begin{proof}
Firstly we note that there is a natural absract convexity on trees \cite{EdelmanJamison1985}: the set is called convex if it is connected.   It is easy to check that it satisfies the {\em Helly property} \cite{Bollobas1986}: if for any family $\mathcal H$ of convex sets in which every two sets have a non-empty intersection we have $\bigcap_{H\in \mathcal H} H \ne \emptyset$.

Let, as above, $V_{xy}$ be the set of voters who prefer $x$ to $y$. These sets are convex.
Consider the set of subsets ${\mathcal M}= \{H_{xy}\mid x\succ y\}$, where $\succ $ is the majority relation. Any two subsets $H_{xy}, H_{zt}\in {\mathcal M}$ have a nonempty intersection since each of them contains a majority of all vertices. By the Helly property we have
$
\bigcap_{H_{xy}\in {\mathcal M}} H_{xy}\ne \emptyset.
$
Voters'  preferences in this intersection coincide with the majority relation.
\end{proof}

In Example~\ref{smallstar} voter 2 is the representative voter.

We need one easy observation.

\begin{lemma}
\label{l1}
Let $\{\row Pn\}$ is a set of distinct linear orders over a set of alternatives $A$, and let $\row kn$ be positive integers. Consider a profile $Q=(P_1^{k_1},\ldots, P_n^{k_n})$, where linear order $P_i$ is repeated $k_i$ times. Then a  profile $P=(\row Pn)$ is single-crossing on a tree iff  $Q=(P_1^{k_1},\ldots, P_n^{k_n})$ is also single-crossing on (another) tree. 
\end{lemma}

\begin{proof}
Suppose $P$ is single-crossing on a  tree $T$ with $n$ vertices $1,2,\ldots, n$  in which we identified $P_i$ with $i$. In the trivial case where $n=1$, $Q=(P_1^{k_1})$ is single-crossing on a line of vertex 1 repeated $k_1$ times, and for all pairs of alternatives the cut is virtual. Otherwise, $n>1$ and we take an arbitrary edge $(i,j)$ now and replace it with the line $(i,i_1,\ldots,i_k,j)$.
If we now associate $k$ copies of $P_i$ with vertices $\row ik$ and do it for every $i$, we claim that $Q$ is single-crossing on this new so obtained tree $T'$. The cuts for the new tree should be as follows: If the edge $(i,j)$ had a cut for some pair of alternatives, then $T'$ should have the cut for the same pair of alternatives  done on the edge $(i_k,j)$ of $T'$.

Suppose now that $Q$ is single-crossing for some tree $T$. For every pair of alternatives $(a,b)$ such that $a\succ_i b$  let us consider a subtree $T_{ab}$, where voters corresponding to vertices all prefer $a$ to $b$. Then  the subtree
$
T_P=\bigcap_{a\succ_ib}T_{ab} 
$
is nonempty and contains all vertices corresponding to $k_i$ copies of $P_i$ in $Q$. Since there are no cuts to edges of this tree we can glue all its vertices together and associate with $P_i$ in $P$. We  leave all cuts in the new tree where they were.
\end{proof}

Let $P$ be a profile. By $\mathcal{D}(P)$ we denote the set of all unique linear orders present in $P$.

\begin{theorem}
\label{cd_theorem}
Let $P$ be a profile which is single-crossing with respect to a tree. Then  $\mathcal{D}(P)$ is a Condorcet domain.
\end{theorem}

\begin{proof}
Let $Q=(\row Qm)$ be a profile with $Q_i\in \mathcal{D}(P)$ for all $i$. By Lemma~\ref{l1} we can add linear orders $Q_{m+1},\ldots, Q_n$ so that the extended profile $ \widehat{Q}=(\row Qn)$ is single-crossing on a tree $T=(\widehat{V},E)$, where $\widehat{V}=\{1,\ldots,n\}$. We will also denote $V=\{1,\ldots,m\}$. Suppose  that  $a\succ b$ and $b\succ c$, where $\succ $ is the majority relation for $Q$. Then we have two partitions of $V$, namely, $V=V_{ab}\cup V_{ba}$ and $V=V_{bc}\cup V_{cb}$, where $V_{xy}$ is the set of voters in $\{1,\ldots,m\}$ who prefer $x$ to $y$. We have $|V_{ab}|>|V_{ba}|$ and $|V_{bc}|> |V_{cb}|$. Obviously, $V_{abc}=V_{ab}\cap V_{bc}\ne \emptyset$ since each of these sets contains a majority of voters. We have $V_{ab}\subseteq \widehat{V}_{ab}$ and $V_{bc}\subseteq \widehat{V}_{bc}$, where $\widehat{V}_{ab}$ and $\widehat{V}_{bc}$ are connected. We cannot claim that $\widehat{V}_{ab}$ or $\widehat{V}_{bc}$ contains more than half of all elements of $\widehat{V}$ but we know that $\widehat{V}_{ab}\cap\widehat{V}_{bc}\supseteq V_{abc}\ne \emptyset$. Since $ac$-cut in $T$ cannot be located between vertices of $\widehat{V}_{ab}\cap\widehat{V}_{bc}$, we have either $\widehat{V}_{ac}\supset \widehat{V}_{ab}$ or $\widehat{V}_{ac}\supset \widehat{V}_{bc}$ which implies $V_{ac}\supset V_{ab}$ or $V_{ac}\supset V_{bc}$. This means $a\succ c$ and $\succ$ is transitive.
\end{proof}

\subsection{All trees  produce Condorcet domains}

For three alternatives we do not get anything new. We omit this proof due to space constraints. We will now prove that for an arbitrary tree $T$ there exists of a profile $P$ which is single-crossing with respect to $T$ and $T$ is minimal. In fact we prove the existence of a profile $P$ over $A$ where $|A|=|V|$, i.e., $m=n$. We note, however, that it is not necessary for the number of voters to be the same as the number of alternatives.

\begin{definition}
Let $P$ be a profile over a set of candidates $A$ and let $T=(V,E)$ be a tree. If an edge $(i,j)\in E$ is not the $ab$-cut for any alternatives $a,b\in A$ then we say that $(i,j)$ is a {\em collapsible edge}. If $T$ contains a collapsible edge then we say $T$ is {\em collapsible} with respect to $P$. Otherwise we say that $T$ is {\em minimal} with respect to $P$. 
We call a profile {\em reduced} if it does not contain identical linear orders.
\end{definition}

\begin{theorem}
Let $T=(V,E)$ be a tree with $|V|=n\geq2$. Then there exists a reduced $n$-voter profile $P$ over a set of alternatives $A=\{a_1,\ldots, a_n\}$ that is single-crossing with respect to $T$ and $T$ is minimal for $P$. Moreover, $n$ is the smallest number of alternatives that allows, for any tree $T$ with $n$ vertices, to construct a profile which is single-crossing with respect to $T$.
\end{theorem}

\begin{proof}
We prove the first part of the theorem by induction on $n$. 
The base case is the only tree with two vertices for which the result is clear. 

Suppose now that for an arbitrary tree $T$ on $n$ vertices there exists
a profile $P=(P_1,P_2,...,P_n)$ over $A=\{a_1,a_2,...,a_n\}$ which is single-crossing with respect to $T$ and every edge of $T$ is the $a_ia_j$-cut for some $a_i, a_j \in A$. 

Now, the inductive step. Consider the case for $T=(V,E)$, a tree with $n+1$ vertices. Remove any vertex of degree 1 (a leaf) from $T$, along with the edge that was connected to that vertex, and call the resulting tree $T'$. We label the vertices so that $n+1$ is removed, and the removed edge is $(n,n+1)$. Now we have $T'=(V',E')$ where $|V'|=n$. Then by the induction hypothesis there exists a reduced profile $P'=(P_1,P_2,\ldots,P_n)$ over a set of alternatives $A'=(a_1,a_2,\ldots,a_n)$ that is single-crossing with respect to $T'$ and $T'$ is minimal. Now add back the removed vertex and edge to recover the tree $T$. We then have $V=V'\cup \{n+1\}$ and $E=E'\cup{(n,n+1)}$. Now we extend the profile $P'$ to $P=(P_1,P_2,\ldots,P_n,P_{n+1})$ where $P_{n+1}$ is  identical to $P_n$, and then extend the set of alternatives $A'$ to $A=(a_1,a_2,...,a_n,a_{n+1})$. We add alternative $a_{n+1}$ to each of $P_1,...,P_n$ immediately after $a_n$ and immediately before $a_n$ in $P_{n+1}$.  Now voters $n$ and $n+1$ do not have identical preferences over $A$, and now we have a set of $n+1$ linear orders over $n+1$ alternatives and the resulting profile is reduced. Since $a_n$ and $a_{n+1}$ appear consecutively in all lists of preferences, $a_n\, \succ_j\, a_i$ if and only if $a_{n+1}\, \succ_j\, a_i$ for any $a_i\in A\setminus\{a_n,a_{n+1}\}$ and $j=1,2,\ldots,n+1$. Therefore the subtree where $a_n$ is rated above $a_i$ is identical to the subtree where $a_{n+1}$ is rated above $a_i$, so the $a_ia_n$-cut is the same as the $a_ia_{n+1}$-cut so the $a_ia_{n+1}$-cut exists for $i=1,2,\ldots, n-1$. Also $a_na_{n+1}$-cut  is the new edge  $(n,n+1)$.  It follows that a cut exists for every pair of alternatives, so $P$ is single-crossing with respect to $T$ and $T$ is minimal.

Hence by induction for any tree $T$ with $|V|=n\geq2$ there exists a profile $P$ over a set of alternatives $A$ that is single-crossing with respect to $T$, where $T$ is minimal.

To prove the second part let us consider the star graph $S_n$ on $n$ vertices whose one vertex has degree $n-1$ and all others are leaves.  Reasoning by induction let us assume that for $S_n$ we cannot construct a profile with less than $n$ alternatives for which $S_n$ is the minimal tree. Consider $S_{n+1}$ and suppose towards a contradiction that we can find a profile $P=(\row P{n+1})$ with $n$ alternatives for which $S_{n+1}$ is minimal. Suppose that the vertices of $S_{n+1}$ are numbered so that voter 1 has degree $n$ and has preferences expressed by the linear order $a_1\succ_1 a_2\succ_1\ldots \succ_1 a_{n}$.  Suppose, first, the $a_1a_2$-cut is not virtual for $S_{n+1}$. Without loss of generality we may assume that $a_2\succ_{n+1} a_1$. Then $a_1\succ_i a_2$ for all $i\le n$. let us remove now vertex $n+1$ from the graph. Then we get graph $S_n$ and the profile $P'=P_{-(n+1)}$ which is single-crossing on $S_n$. The $a_1a_2$-cut  becomes virtual for $P'$. Let us now remove the alternative $a_1$ from $P'$ to obtain a profile $P''$. Then $P''$ is still single-crossing on the tree $S_n$ and we claim that $S_n$ is minimal for $P''$. If not, then after removal of $a_1$, at least one edge, say $(1,k)$ becomes without any cuts. Hence the only cut it had was the cut $a_1a_j$ for some $j>2$. Then we had $a_j\succ_k a_1\succ_k a_2$ and see that $(1,k)$ had also a $a_2a_j$-cut. Thus $S_n$ is minimal for $P''$ and $P''$ has $n-1$ alternative which is a contradiction.
\end{proof}

We will now make a trivial but very useful observation. 

\begin{proposition}
\label{leaves}
Let $P=(\row Pn)$ be a reduced profile which is single-crossing with respect to a minimal tree $T$. Then a voter $i$ is a leaf (vertex of degree 1) of $T$ iff there exist a pair of alternatives $a,b\in A$ such that $a\succ_i b$ and $b\succ_j a$ for all $j\in N\setminus \{i\}$.
\end{proposition}

\begin{proof}
If $i$ is a leaf, then it is connected to the rest of the tree with a single edge, say $(i,k)$. This edge is an $ab$-cut for some $a,b\in A$ from which  $a\succ_i b$ and $b\succ_j a$ for all $j\in N\setminus \{i\}$. Conversely, if for some $a,b\in A$ from which  $a\succ_i b$ and $b\succ_j a$ for all $j\in N\setminus \{i\}$, then the the $ab$-cut partitions $T$ into a certain subtree $T'$ and a single vertex $i$ which means it was a leaf. 
\end{proof}


\begin{theorem}
\label{uniqueness}
Let $P=(\row Pn)$ be a reduced profile over a set of alternatives $A$. Suppose $P$ is single-crossing. Then the minimal tree with respect to which $P$ is single-crossing is unique.
\end{theorem}

\begin{proof}
Suppose for contradiction that $P$ is single-crossing with respect to both $T_1=(V_1,E_1)$ and $T_2=(V_2,E_2)$ and both are minimal. 
We prove that $T_1=T_2$ by induction on $n=|T_1|=|T_2|$.  The case $n=1$ is obvious. Now let $P$ be a profile with $n+1$ voters. Choose a leaf  from $T_1$, say vertex $i$, and let the vertex to which $i$ is connected be $j$. Since it is a leaf it must be the only linear order in which $x$ is preferred over $y$ for some $x,y\in A$ and therefore vertex $i$ must also be a leaf of $T_2$. 

Note that, if $P$ is single-crossing with respect to a tree $T $, and we remove a leaf of $T$, say vertex $\ell$ (together with the edge leading to it), then we obtain a tree $T_{-\ell}$ such that the corresponding subprofile  $P_{-\ell}=(P_1,...,P_{\ell-1},P_{\ell+1},...,P_n)$ is single-crossing with respect to $T_{-\ell}$. 

Therefore if we remove $i$ from $T_1$ and $T_2$ we are left with subtrees $T'_1$ and $T'_2$ and $P_{-i}$ is single-crossing with respect to both. Since $|T'_1|=|T'_2|=n$, by the induction hypothesis we must have $T'=T'_1=T'_2$. Suppose that $(i,j)\in E_1$ and $(i.k)\in E_2$. Then to recover $T_1$ from $T'$ we add the vertex $i$ and the edge $(i,j)$ to $T'$ and to recover $T_2$ from $T'$ we add the vertex $i$ and the edge $(i,k)$  to $T'$ for some $j,k$ different from $i$.  If $j\neq k$ then there must be a path in $T'$ from $j$ to $k$, along which there must be at least one $xy$-cut  since $T$ is minimal for $P$. Then $x\succ_j y$ and $y\succ_k x$ for some $x,y\in A$. Making the $xy$-cut on each of $T_1$ and $T_2$ we find that $i\in V_{xy} \subset V_1$ and $i\in V_{yx} \subset V_2$, so $x\succ_i y$ and $y\succ_i x$, a contradiction. Thus $j=k$, $E_1=E_2$ and finally $T_1=T_2$. So the tree with respect to which $P$ is single-crossing is unique.
%
\end{proof}

We note that the new concept is not hereditary, i.e., not inherited by subprofiles. It is easy to see that all subprofiles of $P$ are single-crossing  if and only if $P$ is  single-crossing in the classical sense.

\begin{theorem}
Let $P$ be a single-crossing profile. Then all subprofiles of $P$ are single-crossing  if and only if $P$ is  single-crossing on a line.
\end{theorem}

\begin{proof}
It follows from the definition that if $P$ is single-crossing on a line, then every subprofile is single-crossing. 

Now suppose that  $P$ is single-crossing with respect to a minimal tree $T$ which is not a line. 
This tree  therefore must have a vertex of degree greater than 2, let this be vertex~$i$ and let $j,k,\ell$ be three of its (arbitrarily chosen) neighbors. Then the subprofile
$P'=(P_j,P_k,P_\ell)$ cannot be single-crossing on any tree. If it were single-crossing on a tree $T$, then by Proposition~\ref{leaves} all three vertices of it must be leaves which is impossible.
\end{proof}

\section{Algorithmic aspects of single-crossedness}

\subsection{A recognition algorithm.}

The main question of this section is to give a polynomial-time algorithm for recognising single-crossing profiles. The property, which was used in Proposition~\ref{leaves}  to recognise a leaf when we knew that the profile had single-crossing property becomes too weak for arbitrary profiles. 
%
%
We need a stronger property.

\begin{definition}
\label{poleaf}
Let $P=(\row Pn)$ be a profile. We call $i\in N$ a {\em potential leaf} if 
\begin{itemize}
\item[(a)] The set $S_i=\{(a,b)\in A^2\mid a\succ_i b\ \text{and}\ b\succ_j a\ \text{for $j\ne i$}\}$ is nonempty;
\item[(b)] There exists $k\in N$ such that $a\succ_i b \Leftrightarrow a\succ_k b$ for all $a,b\in A$ such that neither $(a,b)$ or $(b,a)$ belongs to $S_i$.
\end{itemize}
\end{definition}

Just as a quick reality check we verify that if a profile $P$ is single-crossing with respect to  a tree $T$, then a vertex is a leaf of $T$ if and only if it is a potential leaf.
%
%
So leaves and potential leaves are equivalent for single-crossing profiles. However potential leaves are more useful in the general case.

\begin{lemma}
Let $P=(\row Pn)$ be a profile and $i\in N$ is a potential leaf of $P$. Then $P$ is single-crossing if and only if $P_{-i}$ is.
\end{lemma}

\begin{proof}
Suppose $i$ is a potential leaf and that $P_{-i}$ is single-crossing with respect to a tree $T=(N\setminus \{i\},V)$. According to Definition~\ref{poleaf}, since $i$ is a potential leaf of $P$, all $ab$-cuts for $T$, where either $(a,b)$ or $(b,a)$ belongs to $S_i$, are virtual. We add $P_i$ to $P_{-i}$, add $i$ to $T$ and connect $i$ to $k$. Then $(i,k)$ becomes an $ab$-cut for all $a,b\in A$ for which either $(a,b)$ or $(b,a)$ belongs to $S_i$. All other cuts remain valid. This proves that $P$ is single-crossing with respect to $T'=(N, V\cup \{(i,k)\}$.  

Let us prove the converse. Suppose $P$ is single-crossing and $T=(V,E)$ be its minimal tree. Then $i$ is a leaf of $T$. In general, as we know, the single-crossing property is not inherited by subprofiles. However when we remove a leaf (or several of them), then the single-crossedness is preserved. Indeed, for every pair of alternatives $x,y\in A$ the $xy$-cut is either in $S_i$, and then it becomes a virtual cut in $T_{-i}$, or it is not. In the latter case,  $V=V_{xy}\cup V_{yx}$, and suppose without loss of generality that $i\in V_{xy}$.  Since $i$ was a leaf of $T$ it will also be a leaf of $V_{xy}$, hence $V_{xy}\setminus \{i\}$ is a subtree of $T_{-i}$. Hence $T_{-i}$ has a $xy$-cut and $P_{-i}$ is single-crossing with respect to $T_{-i}$.
\end{proof}

The idea of the recognition algorithm is now clear: we look for a potential leaf and stop if we cannot find any. Otherwise we remove the potential leaf, remember where it was attached and reduce the problem to the remaining subprofile. For a linear order $P_i$ to calculate $S_i$ we need ${m \choose 2}n$ operations, to find $k_i$ we need the same number of comparisons and we have to try in the worst case scenario  all $n$ linear orders so $2{m \choose 2}n^2$ operations in total. We have to do it recursively $n$ times so the total number of operations is at most $2{m \choose 2}n^3$. We have proved

\begin{theorem}
\label{conmintree}
For an input profile with $n$ voters and $m$ alternatives we can determine whether or not this profile is single-crossing or not and, if it is, to construct the minimal tree in $O(m^2n^3)$ time.
\end{theorem}

\subsection{Chanberlin-Courant rule 
}

Given a society of $n$ voters $V$ with preferences over a set of $m$ candidates $A$ and a fixed positive integer $k\le m$ a method of fully proportional representation outputs a $k$-member parliament, which is a subset of $A$, and assigns to each voter a candidate that  will represent this voter in the parliament. The Chanberlin-Courant  and the   Monroe rules were widely discussed in Political Science and Social Choice literature. 

It is well-known that on an unrestricted domain of preferences both rules are intractable in the classical  \cite{mei-pro-ros-zoh:j:multiwinner,bou-lu:c:chamberlin-courant} and parameterized complexity  \cite{bet-sli-uhl:j:mon-cc} senses.  \cite{SYFE2013}, however, showed that for single-crossing  elections the winner-determination problem for the Chanberlin-Courant rule is polynomial-time solvable for every dissatisfaction function and both for the utilitarian  \cite{cha-cou:j:cc} and egalitarian \cite{bet-sli-uhl:j:mon-cc} versions of the rule. They also generalized this result to elections with bounded single-crossing width proving fixed-parameter tractability of the Chanberlin-Courant  rule with single-crossing width as parameter.  The concept of single-crossing width was defined in \cite{CornazGS12}. 

Here we will prove that polynomial solvability remains for single-crossing profiles on any tree. But, firstly, we will remind the reader the definitions needed for discussing the Chanberlin-Courant rule. By $\text{pos}_v(c)$ we denote the position of the alternative $c$ in the ranking of voter $v$; the top-ranked alternative has position 1, the second best has position 2, etc.

\begin{definition}
\label{def:mf}
Given an $n$-profile $P$ over set~$A$  of alternatives, a mapping $r\colon P\times A \to {\Q}^+_0$
is called {\em a misrepresentation function} if for any voter $v\in N$
and any two candidates $c,c'\in A$ the condition
$\text{pos}_v(c)<\text{pos}_v(c')$ implies $r(v,c)\le r(v,c')$. 
\end{definition}
In the classical framework the misrepresentation of a candidate
for a voter is a function of the position of the candidate in the preference
order of that voter given by ${\bf s}=(\row s{m})$, where $0=s_1\le s_2\le \ldots \le s_m$, that is, the misrepresentation function in this case will be
$
r(v,c)=s_{\text{pos}_v(c)}.
$
An important particular case is the \emph{Borda misrepresentation  function} defined by the vector
 $(0,1,\ldots, m-1)$ which was used in  \cite{cha-cou:j:cc}. 
 %
 In what follows we assume that the misrepresentation function is defined for any number of alternatives and that it is polynomial-time computable.

In the approval voting framework, if a voter is represented by a
 candidate whom she approves, her misrepresentation is  zero, otherwise it is equal to one.  This function is called
 the \emph{approval misrepresentation function}.  This
 misrepresentation function does not have to be positional since
 different voters may approve different number of candidates. 
In the general framework the misrepresentation function may be
arbitrary.  

By $w\colon N\to A$ we denote the function that assigns
voters to representatives (or the other way around), i.e., under this
assignment voter $v$ is represented by candidate $w(v)$. If $|w(N)|\le k$ we call it a $k$-assignment. The total
misrepresentation $\Phi(P,w)$ of the given election under~$w$ is then given by
$
\Phi(P,w)=\sum_{v \in N} r(v,w(v))\quad \text{or}\quad \Phi(P,w)= \max_{v \in N} r(v,w(v))
$
in the utilitarian and egalitarian versions,
respectively. The Cham\-ber\-lin-Courant rule takes the profile and the number of representatives to be elected $k$ as input and outputs an optimal $k$-assignment  $w_\text{opt}$ of voters to representatives that minimizes  the total misrepresentation $\Phi(P,w)$.  We will prove the following theorem.

\begin{theorem}
\label{CCtheorem}
For every polynomial-time computable dissatisfaction function, every positive integer $k$, and for both utilitarian and egalitarian versions of the Chanberlin-Courant rule, there is a polynomial-time algorithm that given a profile $P=(\row Pn)$ over a set of alternatives $A=\{\row am\}$, which is single-crossing with respect to some tree,  finds an optimal $k$-assignment function $w_\text{opt}$ for $P$.
\end{theorem}


Let $T$ be a tree. Any subtree $T'$ of $T$ such that $T\setminus T'$ is also a tree be called a {\em terminal subtree} of $T$. Let us denote by $st(T)$ the set of all terminal subtrees of $T$, including $T$.
%
An important example of a terminal subtree can be obtained from any $ab$-cut. Indeed, both subtreees $T_{ab}$ and $T_{ba}$, which result from this cut are terminal. Let us, for $i\in N$, denote by $st_{-i}(T)$ the set of all terminal subtrees for which $i\in N$ is not a vertex.

\begin{lemma}
\label{termsubtr}
Let $P=(\row Pn)$ be a profile over a set~$A=\{\row am\}$ of alternatives, where $P$ is single-crossing with respect to a minimal tree $T$. Let $w_\text{opt}$ is an optimal $k$-assignment for $P$.  Let voter 1 be an arbitrary vertex of $T$, and let $b\in A$ be the least preferred alternative of voter 1 in $w(N)$. Then the vertices of $w^{-1}(b)$ are vertices of a terminal subtree of $T$.
\end{lemma}

\begin{proof}
On a tree $T$ we may define the distance between any two vertices $u$ and $v$ which is the number of edges on the unique path connecting these two vertices. 
Let $v$ be the closest vertex to $1$ such that $w(v)=b$. Let $u$ be the vertex on that path which is one edge closer to $1$ (it can be actually $1$ itself). Due to minimality of $T$ the edge $(u,v)$ is an $(a,b)$-cut for some $a= w(u)$ where $a\succ_1 b$. Let us show that $w^{-1}(b)=V_{ba}$. Suppose first that $w(v')=b$. Then $b\succ_{v'} a$ (otherwise $v'$ would be assigned $a$) and hence $v'\in V_{ba}$. Suppose now $v'\in V_{ba}$. Then $v'$ is connected by a path within $V_{ba}$, hence the unique path between 1 and $v'$ passes through $v$. Since linear orders on this path 
form a classical single crossing subprofile we have $w(v')=b$.
\end{proof}

The fact that the least preferred alternative of voter 1 in the elected committee $w(N)$ represents voters in a terminal subtree of $T$ is important in our design of a dynamic programming algorithm. In this algorithm we will fix an arbitrary leaf, without loss of generality it will be voter 1, and reduce the problem of calculating an optimal assignment for a profile $P=(\row Pn)$ over a set of alternatives $A=\{\row am\}$ to a partially ordered subproblems which will be defined shortly. Let us denote the original problem as $(P,A,k)$. We define the set of subproblems as follows. A pair $(P',A',k')$, where $k'\le k$ is a positive integer, $P'$ is a subprofile of $P$ and $A'$ is a subset of $A$ is a subproblem of $(P,A,k)$ if 
\begin{enumerate}
\item $P'=P_{ab}$, the subprofile of linear orders corresponding to vertices of a subset of vertices $V_{ab}\subseteq N$, where $a\succ_1 b$;
\item $A'=\{\row a{j}\}$ for some positive integer $j$ such that $k\le j\le n$;
\item the goal is to find an optimal assignment $w'\colon V_{ab}\to A$ with $|w'(V_{ab})|\le k'$.
\end{enumerate}

The idea behind this definition is based on Lemma~\ref{termsubtr}.   Namely, we can try to guess the least preferred alternative of voter 1 in the elected committee $a_j$ and the terminal subtree $V_{ba}$ whose voters are all assigned to $a_j$, then the problem will be reduced to choosing a committee of size $k-1$ among $\{\row a{j-1}\}$ given the profile $P_{ab}$ of voters corresponding to $V_{ab}$. We note that there are at most $n-1$ subproblems. Note that our cuts can be naturally ordered: we can say that $ \text{$ab$-cut}\subset \text{$cd$-cut}$ if and only if $V_{ab}\subset V_{cd}$.\par\smallskip

 Following \cite{SYFE2013} we can prove the main theorem of this section. \par\smallskip


\noindent{\it Proof of Theorem~\ref{CCtheorem}.}
The tree $T$ with respect to which $P$ is single-crossing can be computed in polynomial time, by Theorem~\ref{conmintree}. We can identify one of the leaves then. Let this be voter~1 with preferences  $a_1\succ_1...\succ_1a_m$.

For every $ab$-cut, $j\in \{1,\ldots,m\}$ and $t\in \{1,\ldots,k\}$ we define $A[V_{ab},j,t]$ to be the optimal  dissatisfaction (calculated in the utilitarian or egalitarian way) that can be achieved with a $t$-assignment function  when considering a subprofile $P'=\{P_{ab}\,|\,1\in V_{ab}\}$ over $A'=\{\row aj\}$. It is clear that for $V\subseteq N$, $j\in M\setminus \{1\}$ and $t\in K\setminus \{1\}$, the following recursive relation holds
\begin{align*}
A[V,j,t]=&\min\left\{A[V,j-1,t],\min\limits_\text{$ab$-cut}\ell(A[V_{ab},j-1,t-1],\right. \\
&(r(v_i,a_j))_{v_i\in V_{ba}})\Big\}.
\end{align*}
Here $\ell$ is used to mean either the sum or the maximum of a given list of values  depending on the utilitarian or egalitarian model, respectively. To account for the possibility that $a_j$ is not elected in the optimal solution, we also include the term $A[V, j - 1, t]$. Thus,
$A[N,m,k]$ is then the optimal dissatisfaction, and in calculating it we simultaneously find the optimal $k$-assignment function.

The following base cases are sufficient for the recursion to be well-defined
\begin{itemize}
\item $A[\varnothing,j,t]=0$;
\item $A[V,j,1]=\min\limits_{j'\leq j} \ell((r(v_i,a_{j'}))_{i\in V})$;
\item $A[V,j,t]=0$ for $t\geq j$.  
\end{itemize}

Using dynamic programming, we can compute in polynomial time, in fact, in time $O(mn^2k)$, the optimal
dissatisfaction of the voters and the assignment that achieves it. Note that the complexity is the same as for the classical case.

\section{Conclusion}

This paper generalises the classical single-crossing property to a single-crossing property on trees. In this more general setting, we prove various results such as the fact that the majority relation is transitive as well as a representative voter theorem. Furthermore, for any tree, there exists a profile of preferences which is single-crossing with respect to the tree and for which the tree is minimal. Finally we present two results on algorithmic aspects of single-crossedness. The first one states that recognising single-crossingness on a tree is polynomial and the second shows that the winner determination problem for the Chamberlin-Courant rule is also polynomial for single-crossing profiles.

\bibliographystyle{elsarticle-harv}
\bibliography{Single-crossing}     

\end{document}